\documentclass[a4paper, 11pt]{article}

\usepackage[utf8]{inputenc}
\usepackage{stmaryrd}
\usepackage{ifthen}
\usepackage{amsfonts, amsmath, amssymb, amsthm, bbm}
\usepackage[dvipsnames]{xcolor}
\usepackage{tikz}
\usetikzlibrary{arrows,shapes,snakes,automata, calc, chains, matrix, positioning, scopes}
\usepackage{paralist}
\usepackage{geometry}
\usepackage{xcolor}
\usepackage{xspace}
\usepackage{authblk}

\newtheorem{theorem}{Theorem}

\newtheorem{lemma}{Lemma}

\newtheorem{example}{Example}
\newtheorem{remark}{Remark}

\theoremstyle{plain}

\newcommand{\alphabet}{\Sigma}

\newcommand{\A}{\mathcal A}
\newcommand{\B}{\mathcal B}
\newcommand{\M}{\mathcal M}
\newcommand\G{\mathcal G}
\newcommand{\N}{\mathbbm N}
\newcommand{\Z}{\mathbbm Z}
\newcommand{\Q}{\mathbbm Q}
\newcommand{\R}{\mathbbm R}
\newcommand{\Ns}{\mathbbm N\setminus\{0\}}
\newcommand{\Zm}{\mathbbm{Z}_{\max}}
\newcommand{\Nm}{\mathbbm{N}_{\max}}
\newcommand{\matrice}[3]{\M_{#1,#2}(#3)}
\newcommand{\matrices}[2]{\M_{#1}(#2)}

\newcommand{\dfn}[1]{\emph{#1}}
\newcommand{\resp}[1]{\ (resp. #1)}
\newcommand{\z}{-\infty}

\newcommand{\fonc}[1]{\left\llbracket #1\right\rrbracket}
\newcommand{\norm}[2][]{\left\Vert #2\right\Vert_{#1}}

\newcommand{\length}[1]{\left|#1\right|}
\newcommand{\semi}[1]{\left\langle #1 \right\rangle}
\newcommand{\sett}[2]{\left\{\left.#1\vphantom{#2}\right|#2\right\}}

\newcommand\JSR[3][]{\texttt{JSR}\ifthenelse{\equal{#1}{}}{}{_{#1}}^{#2}(#3)}
\newcommand\UR[3][]{\texttt{UR}\ifthenelse{\equal{#1}{}}{}{_{#1}}^{#2}(#3)}
\newcommand\Comp[3][]{\texttt{Pos}\ifthenelse{\equal{#1}{}}{}{_{#1}}^{#2}(#3)}

\renewcommand{\leq}{\leqslant}

\newcommand{\TODO}[2]{{\ifthenelse{\equal{#1}P}{\color{blue}Pierre}{\ifthenelse{\equal{#1}L}{\color{green}Laure}{\ifthenelse{\equal{#1}G}{\color{red}Glenn}{\color{purple}#1}}} TODO: #2}}

\newcommand{\ie}{\textit{i.e.}\ }

\DeclareMathOperator*{\urk}{urk}

\begin{document}

\title{Comparison of max-plus automata and \\ joint spectral radius of tropical matrices}

\author[1]{Laure Daviaud\footnote{The first author was partly supported by ANR Project ELICA ANR-14-CE25-0005, by ANR Project RECRE ANR-11-BS02-0010 and by project LIPA that has received funding from the European Research Council (ERC) under the European Union{'}s Horizon 2020 research and innovation programme (grant agreement No 683080).}}
\author[2]{Pierre Guillon}
\author[2]{Glenn Merlet}
\affil[1]{University of Warsaw \\ Poland}
\affil[2]{Universit\'e d'Aix-Marseille CNRS, Centrale Marseille, I2M, UMR 7373, 13453 Marseille, France}
\date{}                     
\setcounter{Maxaffil}{0}
\renewcommand\Affilfont{\itshape\small}

\maketitle

\begin{abstract}
Weighted automata over the tropical semiring $\Zm = (\Z\cup\{-\infty\},\max,+)$ are closely 
related to finitely generated semigroups of matrices over $\Zm$. 
In this paper, we use results
in automata theory to study two quantities associated with sets of matrices: the joint 
spectral radius and the ultimate rank. We prove that these two quantities are 
not computable over the tropical semiring, \ie there is no algorithm that 
takes as input a finite set of matrices $\Gamma$ and provides as output
the joint spectral radius \resp{the ultimate rank} of $\Gamma$. 
On the other hand, we prove that the 
joint spectral radius is nevertheless approximable and we exhibit restricted cases in which 
the joint spectral radius and the ultimate rank are computable. 
To reach this aim, we study the problem of comparing functions computed by 
weighted automata over the tropical 
semiring. This problem is known to be undecidable, and we prove that it remains 
undecidable in some specific subclasses of automata.
\end{abstract}

%
%

\section{Introduction}

Weighted automata were introduced by Schützenberger in \cite{Schutz61}
as a quantitative extension of nondeterministic finite automata. They
compute functions from the set of words over a finite alphabet to 
the set of values of a semiring, allowing one to model quantities such as costs, gains or 
probabilities. In this paper, we particularly focus on
max-plus automata: automata weighted within the tropical semiring 
$\Zm = (\Z\cup\{-\infty\},\max,+)$. 
A max-plus automaton is thus a nondeterministic finite automaton 
whose transitions are weighted by integers.
The value associated to a word $w$ depends on the runs labelled by $w$:
the weight of a given run is the sum of the weights of the transitions
in the run, 
and the weight of $w$ is the maximum of the weights of the accepting runs
labelled by $w$. This kind of automata is particularly suitable 
to model gain maximisation, to study worst-case complexity~\cite{ColDav14} 
and to describe discrete event systems~\cite{Gaub95, GaubMair99}.
The so-called linear presentation gives a matrix representation
of such an automaton. More precisely, there is a canonical way to associate
a max-plus automaton with a finitely generated semigroup of matrices 
over $\Zm$. Usually, the matrix representation is used to
provide algebraic proofs of automata results. In this paper, we use results in automata
theory to study two quantities related to sets of matrices: the joint spectral radius
and the ultimate rank. The joint spectral radius
generalises the notion of 
spectral radius for sets of matrices.
The ultimate rank unifies some usual other notions of ranks.
We link some comparison problems on max-plus automata
with the computation of these two quantities.
This leads to (1) prove results about comparison problems in 
some restricted classes of max-plus automata that
we believe to be interesting for themselves and (2) apply these results 
to the study of the computability of the joint spectral radius and the ultimate rank.

\paragraph*{Comparison of max-plus automata}
Decidability questions about the description of functions computed by 
max-plus automata have been intensively studied. 
In his celebrated paper~\cite{Krob92}, Krob proves the undecidability of the 
equivalence problem for max-plus automata: 
there is no algorithm to decide if two max-plus automata compute 
the same function.
In fact, his proof gives a stronger result: it is undecidable to determine
whether a max-plus automaton 
computes a positive function.
A more recent proof of this result, due to Colcombet, is based on a reduction
from the halting problem of two-counter machines ~\cite{Colcombet}.
By various reductions, this leads to the undecidability of several properties 
of automata with weights in different versions of the tropical semiring:
$(\N\cup\{-\infty\}, \max, +)$, $(\N\cup \{+\infty\}, \min, +)$...
The reader is referred to~\cite{LombardyMairesseSurvey} for a survey on these questions.

\paragraph*{Restriction on the parameters}
From the proof through two-counter machines, it can be derived that the undecidability remains even if the 
automata are restricted to have weights within $\{-1,1\}$.
In \cite{GaubertKatz2006}, Gaubert and Katz notice that the undecidability of the 
comparison also remains true 
even if the number of states of the automata is bounded by a certain integer $d$. This
extension is
based on Krob's original proof and on the use of a universal diophantine equation. 
However, they ask for a more direct proof that would allow one to control the bound $d$.
As an attempt to answer this question, we extend the proof through two-counter machines.
This allows a much sharper bound on the number of states
for which comparison is undecidable (Theorem~\ref{theorem:main})
than the one that would have followed from a universal diophantine equation.

\paragraph*{Restriction on the initial and final states}
The class of functions computed by max-plus automata that have all their states both initial and final 
is strictly included in the class of functions computed by max-plus automata.
However, it is closely related to the study of finitely generated semigroup of tropical matrices.
In this paper, we prove that comparison remains undecidable in this restricted class.
This involves a reduction from the general case, that was quite surprisingly never noticed before (Theorem~\ref{theorem:if}).

The result in~\cite{Almagor2011} is an alternative proof that the comparison is undecidable for min-plus automata with weights in \{-1,0,1\} and all states final 
(and we can deduce the same result for max-plus automata). 
Our result is stronger in that it proves that the comparison is undecidable for max-plus automata with weights in \{-1,0,1\} and all states both initial and final, 
as well as max-plus automata with a bounded number of states. Moreover our proof (and already Colcombet’s proof) constructs a polynomially-ambiguous max-plus automaton, 
proving that the undecidability still holds for this restricted class, that is not clear from the proof in~\cite{Almagor2011}.

\paragraph*{Joint spectral radius and ultimate rank}
Although the joint spectral radius is a well-studied notion when considering 
the semiring $(\R,+,\times)$ (see \cite{JSRBookJungers} and the references therein)
only few results are known when considering
the tropical semiring. 
As far as we know, the best known result is given in \cite{RhoMinNPHard}, where it is shown that the joint 
spectral radius is NP-hard to compute and to approximate for tropical matrices.
We drastically improve these results by proving that the joint spectral radius 
is not computable in the tropical semiring, \ie there is no algorithm that 
takes as input a finite set $\Gamma$ of matrices and provides as output
the joint spectral radius of $\Gamma$ (Theorem~\ref{t:fixedw}). 
As a corollary of this result, we also get the uncomputability of the 
ultimate rank, a notion introduced -- and a question raised -- in \cite{urk} (Theorem~\ref{t:ur}).

On the other hand, we also give positive results. 
By making a link with a result in \cite{ColDav13} about approximate comparison 
of max-plus automata, we 
prove that the joint spectral radius is approximable in EXPSPACE
(Theorem~\ref{theorem:approx}).
We also show that, when restricted to matrices with only finite rational entries,
the joint spectral radius and the ultimate rank can be computed 
in PSPACE (Theorem~\ref{t:FiniteEntries}).

\paragraph*{Organisation of the paper}
In Section~\ref{section:def}, we give the definitions and useful notions of 
tropical matrices and max-plus automata.
In Section~\ref{section:comparison}, we discuss about 
the undecidability of comparison of max-plus automata.
In Section~\ref{section:jsr}, we give the definitions of the joint spectral radius and ultimate rank and prove that these two quantities are uncomputable.
We also give positive results 
about the decidability of the approximation of the joint spectral radius
and its computation in restricted cases.
Section~\ref{a:comparaison} is devoted to the technical embedding of counter machines in weighted automata, useful for the undecidability of comparison. 

%
%

\section{Definitions and first properties}
\label{section:def}

In this section, we introduce definitions and notation of tropical matrices and max-plus automata.

\subsection{Tropical matrices}

A \dfn{semigroup} $(S,\cdot)$ is a set $S$ equipped with an
associative binary operation~`$\cdot$'.
If, furthermore, the product has a neutral element $1$, 
$(S,\cdot,1)$ is called a \dfn{monoid}.
The monoid is said \dfn{commutative} if $\cdot$ is commutative.
A \dfn{semiring} $(S,\oplus,\otimes,0_S,1_S)$ is a set $S$ equipped with 
two binary operations $\oplus$ and $\otimes$
such that $(S,\oplus,0_S)$ is a commutative monoid,
$(S,\otimes,1_S)$ is a monoid, $0_S$ is absorbing for $\otimes$,
and $\otimes$ distributes over $\oplus$.
We shall use the \dfn{tropical semiring}:
$$\Zm = (\Z\cup\{-\infty\},\max,+,-\infty,0)$$ 
Remark that $0_{\Zm}=-\infty$ and $1_{\Zm}=0$. We may also use
the restriction of $\Zm$ to the nonnegative integers, 
$(\N\cup\{-\infty\},\max,+,-\infty,0)$ denoted by $\Nm$.

\paragraph*{Semigroups of matrices} 
Let $S$ be a semiring. The set of matrices with $d$ rows  and 
$d'$ columns over $S$ is denoted $\matrice{d}{d'}{S}$, 
or simply $\matrices{d}{S}$ if $d=d'$. 
The set of all matrices over $S$ is $\M(S)$.
As usual, the product $A B$ for two matrices $A,B$
(provided the width of $A$ and the height of $B$ coincide, denoted here $d$) 
is defined as:
\begin{align*}
(A B)_{i,j} & = \bigoplus_{1\leq k \leq d}(A_{i,k}\otimes B_{k,j}) \\
& = \max_{1 \leq k \leq d}(A_{i,k}+B_{k,j}) 
\quad \text{for } S=\Zm
\end{align*}

The diagonal matrix
with $1_S$ (\ie, $0$ for $\Zm$) on the diagonal, 
and $0_S$ (\ie, $-\infty$ for $\Zm$) elsewhere is denoted $I_d$.
It is standard that $(\matrices{d}{S},\cdot,I_d)$ is a monoid. 

For a positive integer $k$, we use the notation 
$M^k=\underbrace{M\otimes\dotsm \otimes M}_{k \text{ times}}$. 
Moreover, $\norm[\infty]{M}$ denotes the maximal entry 
of a matrix~$M$ (it is not a norm).
For $k\in\Zm$ and $A\in\M(\Zm)$, $k\odot A$ is defined by $(k\odot A)_{ij}=k+A_{ij}$. 
For a set of matrices $\Gamma$, this notation is extended by $k\odot\Gamma=\sett{k\odot A}{A\in\Gamma}$. 
Finally, if $\Gamma\subset\matrices dS$, we note $\semi\Gamma$ the 
submonoid generated by $\Gamma$.

\paragraph*{Graph of a matrix}
Any square matrix $A\in\matrices d\Zm$, for $d$ a positive integer, 
can be represented by a graph $\G(A)$: the vertices are the indices $1,\ldots,d$, and there 
is an edge from $i$ to $j$, labelled $A_{i,j}$, if and only if the latter is finite.
The \dfn{spectral radius} $\rho(A)$ of a square matrix $A\in \matrices{d}{\Zm}$, for some 
positive integer $d$, known to be the limit $\lim_{n\to +\infty}\frac1n\norm[\infty]{A^n}$, can be seen as 
the maximal average weight of the cycles in $\G(A)$: 
\[\rho(A)=\max_{\substack{\ell\in\N\setminus\{0\} \\ 1 \leq i_1,
\ldots,i_\ell \leq d}}\left(\frac1\ell 
A_{i_1,i_2,\ldots,i_\ell}\right)\]
where $A_{i_1,i_2,\ldots,i_\ell}$ denotes the sum: $$A_{i_1,i_2}+A_{i_2,i_3}+\ldots+A_{i_{\ell-1},i_\ell}+A_{i_\ell,i_1}$$
The \dfn{critical graph} $\G_c(A)$ is the union of cycles ($i_1,\ldots,i_\ell)$ that achieve 
this maximum.
Its \dfn{strongly connected components} are the maximal sets of vertices $C \subseteq \G_c(A)$ 
such that for any $i,j\in C$ there is a path from $i$ to $j$ in $\G_c(A)$.
The \dfn{cyclicity of a strongly connected component} is the greatest common divisor of 
the length of its cycles. The \dfn{cyclicity of $\G_c(A)$} is the lowest common multiple 
of the cyclicities of its strongly connected components.

%
%

\subsection{Max-plus automata}
\label{subsection:max-plus-def}

We give the definition of max-plus automata that can be viewed as graphs or as
sets of matrices.

A \dfn{max-plus automaton}
$\A$ over the alphabet $\alphabet$ with $d$ states
is a map $\mu$ from $\alphabet$ to $\matrices{d}{\Zm}$ together with 
an initial vector $I \in \matrice{1}{d}{\{0,-\infty\}}$
and a final vector $F \in \matrice{d}{1}{\{0,-\infty\}}$\footnote
{
Note that, unlike variants in the literature, our weighted automata have no input or output weight
(that is, $I$ and $F$ have entries in $\{0,-\infty\}$), but 
this does not restrict the set of computed functions.
}. 
The map $\mu$ is uniquely extended into a morphism, also denoted $\mu$, from the semigroup $\alphabet^+$ of nonempty finite words over alphabet $\alphabet$
into $\matrices{d}{\Zm}$.
The \dfn{function computed by the automaton}, $\fonc\A$, maps each word $w\in\alphabet^+$
to $I \mu(w) F \in \Zm$. 
Sometimes, $0$ will denote the function constantly equal to $0$, and $\ge$ the induced partial order over functions $\alphabet^+\to\Zm$ (so that we can write things like $\fonc{\A} \geq 0$). 

Another way to represent a max-plus automaton is 
in terms of graphs.
Given a map $\mu$ from $\alphabet^+$ to $\matrices{d}{\Zm}$, the corresponding 
automaton has $d$ states $q_1, \ldots q_d$, that correspond to the
lines, or to the columns of the matrices. There is a transition from
$q_i$ to $q_j$ labelled by a letter $a\in \alphabet$, with weight $\mu(a)_{i,j}$, if and only if the latter is finite.
The initial\resp{final} states are the states $q_i$ such that $I_i=0$\resp{$F_i=0$}.
A run over the word $w$ is a path (a sequence of compatible transitions) 
in the graph, labelled by $w$. Its weight is the sum of the weights of the transitions. 
The weight of a given word $w$ is the maximum of the weights of the accepting runs 
(runs going from an initial state to a final state)
labelled by $w$. 
The weight of $w$, given by the graph representation, is exactly the value 
$I \mu(w) F$, given by the matrix presentation. 

Given a positive integer $d$ and a max-plus automaton $\A$ defined by some map $\mu : \alphabet \to \matrices d\Zm$, 
we note $\Gamma_\A=\sett{\mu(a)}{a\in \alphabet}$.
Then
the set of weights on the transitions of $\A$
corresponds to the
finite entries appearing in matrices of $\Gamma_\A$.

\begin{example}
Figure~\ref{figure:max-plus-automaton} gives the matrix and graph 
presentations of a max-plus automaton with $2$ states, both initial (ingoing arrow) 
and final (outgoing arrow),
over the alphabet $\{a,b\}$. 
The function that it computes associates a word $w$ to the value $\max(|w|_a,|w|_b)$
where $|w|_x$ denotes the number of occurrences of the letter $x \in \{a,b\}$ in $w$.
\end{example}

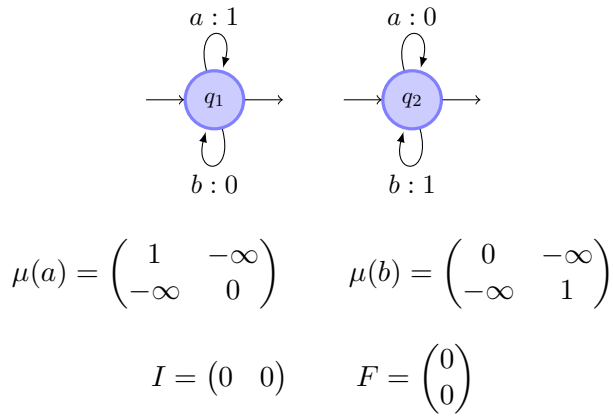
\begin{figure}[!htbp]
\begin{center}
\begin{tikzpicture}
\tikzset{
	every state/.style={draw=blue!50,very thick,fill=blue!20,scale=0.8},
	initial text=,
	accepting/.style={accepting by arrow},
	fleche/.style={->, >=latex}
}

\node[state, initial, accepting] (q_1) at (-1.3,0) {$q_1$};
\node[state, initial, accepting] (q_2) at (1.3,0) {$q_2$};

\path[fleche]     (q_1) edge [loop above] node         {\small{$a:1$}} ();
\path[fleche]     (q_1) edge [loop below] node         {\small{$b:0$}} ();
\path[fleche]     (q_2) edge [loop above] node         {\small{$a:0$}} ();
\path[fleche]     (q_2) edge [loop below] node         {\small{$b:1$}} ();
\end{tikzpicture}
\end{center}

\begin{center}
$\mu(a) = \begin{pmatrix} 1 & -\infty  \\ 
                         -\infty & 0 \end{pmatrix}$ \quad \quad
$\mu(b) = \begin{pmatrix} 0 & -\infty  \\ 
                         -\infty & 1  \end{pmatrix}$ \\
\vspace{0.4cm}
$I=\begin{pmatrix} 0 & 0  \end{pmatrix}$ \quad \quad
$F=\begin{pmatrix} 0\\ 
                   0 \end{pmatrix}$
\end{center}
\caption{\label{figure:max-plus-automaton}Graph and matrix representations of 
a max-plus automaton.}
\end{figure}

This work aims to link
results in automata theory with the study of semigroups of matrices.
Concepts defined over semigroups of matrices correspond to concepts over the subclass of automata in which all states are both initial and final, because if $M\in\matrices d{\Zm}$, and $I$ and $F$ have only $0$ entries, 
 then $IMF=\norm[\infty]M$, so that for this class of automata, $\fonc\A(w)=\norm[\infty]{\mu(w)}$ for every word $w$.

%
%

\section{Undecidability of the comparison of max-plus automata}
\label{section:comparison}

We are interested in the comparison problem, \ie deciding, given two max-plus automata $\A$ and $\B$,
whether $\fonc{\A} \leq \fonc{\B}$.

There exist (at least) two different proofs that this problem is undecidable. The original one by Krob \cite{Krob92} 
is a reduction from the tenth problem of Hilbert about diophantine equations. 
The proof is nicely written in \cite{LombardyMairesseSurvey},
where it encodes a homogoneous polynomial 
$P$ of degree $4$ on $n$ variables with integer coefficients into a max-plus automaton $\A$ computing 
a function with values in $\N$, such that
$P-1$ has a root in $\N^n$ if and only if there is a word $w$ such that $\fonc{\A}(w)=0$, \ie if $\fonc\A\ge1$.

A more recent proof by Colcombet \cite{Colcombet,Almagor2011}, is a reduction from the halting problem
of a two-counter machine. This computational model
was introduced by 
Minsky \cite{Minsky61, Minsky67}, and is as powerful as a Turing machine. They can be viewed as 
finite state machines with two counters that can be incremented,
and decremented if not $0$, 
and the idea of the proof is to embed them into max-plus automata. 

\subsection{Restriction on the parameters}

Different parameters can be taken into account when dealing with the size of a max-plus automaton:
we will focus on the number of states, the maximal and minimal weights appearing on the transitions and the size of the
alphabet.
When considering the matrix representation of an automaton $\A$, these parameters correspond respectively 
to the dimension, the maximal and minimal finite entries and the number of matrices in $\Gamma_\A$.

Regarding the size of the alphabet, by a classical encoding from an arbitrary alphabet to 
a two-letter alphabet, one can see that the comparison problem remains undecidable
when restricting to the class of automata on the binary alphabet.

Regarding the two other parameters, if they are both bounded, the problem becomes decidable
since we are now only considering a finite number of max-plus automata.
What is more interesting to study is when one of the parameter is bounded and not the other.
We will see that this problem remains undecidable in these cases, and our purpose is to give
bounds on these parameters that allow to keep the undecidability.

On the one hand, Gaubert and Katz notice in~\cite{GaubertKatz2006} that 
the original proof of Krob, applied to some specific diophantine equations gives that the problem
remains undecidable when bounding the number of states. They also raised 
the question of finding an alternative proof that could allow to control this number 
of states. We roughly counted how many states we would obtain by using a 
so-called universal diophantine equation given in \cite{jones82}, of degree $4$ with $58$ unknowns.
At the very 
least, we would be able to bound the number of states by $8700$
on a $6$-letter alphabet. 


On the other hand, the proof via two-counter machines allows to drastically improve this number, as we are going to see.

Define $\Comp kS$ \resp{$\Comp[d]kS$} as the following problem:
\textit{Given a max-plus automaton $\A$ on a $k$-letter alphabet, with weights in $S \subseteq \Zm$ 
\resp{and $d$ states}, determine whether $\fonc\A\ge0$.}

\begin{theorem}
\label{theorem:main}
Problems $\Comp2{\{-1,1\}}$ and $\Comp[553]6\Zm$ are 
 undecidable.
\end{theorem}
The first statement is derived rather directly from the proof in \cite{Colcombet,Almagor2011}:
In the proof, weights in $\{-1,0,1\}$ are used, but it is easy to see that with an encoding of the alphabet, every transition with weight $0$ can be replaced by a transition with weight $-1$ followed by one with weight $1$.
This set of weights is clearly minimal.

The undecidability of Problem $\Comp[553]6\Zm$ is a contribution of the present paper.
We extend the construction so that it can simulate two-counter machines on any input (the initial values of the counters).
The input $n\in\N$ is now encoded by an additional widget involving two edges with weights $n$ and $-n$ (it is clear that if the weights were also bounded, the problem would be decidable).
This allows to reduce the halting problem of a universal two-counter machine.
The full proof is given in Section \ref{a:comparaison}.

\subsection{Restriction on initial and final states}

The class of functions computed by max-plus automata having all their states initial and final
is a strict subclass of the functions computed by max-plus automaton.
However, the following lemma shows that the comparison problem remains undecidable in this subclass.

\begin{lemma}
\label{lemma:if}
Let $\alphabet$ be a finite alphabet and $\star \notin \alphabet$ be a special symbol.
Given a max-plus automaton $\A$ on $\alphabet$, 
with $d$ states and weights within a set $S$,
one can build a max-plus automaton $\A'$ on $\alphabet'=\alphabet\cup \{\star\}$, with $d+1$ states, all of which are initial and final, and 
weights within $S \cup \{0\}$,
such that:
\[\min(\inf_{u\in\alphabet^+}\frac{\fonc\A(u)}{\length u},0)\le\inf_{w\in{\alphabet'}^+}\frac{\fonc{\A'}(w)}{\length w}\]
and
\[\inf_{u\in{\alphabet'}^+}\frac{\fonc{\A'}(u)}{\length u}\le\inf_{w\in\alphabet^+}\frac{\fonc\A(w)}{\length w+1}\]
In particular,
$\fonc{\A} \geq 0$ if and only if $\fonc{\A'} \geq 0$.
\end{lemma}

\begin{proof}
Consider a max-plus automaton $\A$ defined by a map $\mu: \alphabet \to \matrices{d}{\Zm}$,
an initial vector $I$ and a final vector $F$, and a new symbol $\star$. Let $\alphabet' = \alphabet \cup \{\star\}$.
The idea is to construct a new automaton $\A'$ by adding a new state $q$ and
transitions from every final state of $\A$ to every initial state of $\A$
as well as transitions from every final state of $\A$ to $q$, loops around $q$ and transitions from $q$ to every initial state of $\A$, all labelled by $\star$ with weight $0$. 
All the states of the new automaton $\A'$ are initial and final. Let us note $\mu'$, $I'$ and $F'$ defining this new automaton.


Any word $w\in \alphabet'^+ \backslash \{\star\}^*$ can be written:
$$w=\star^{n_0}w_1\star^{n_1} w_2 \star^{n_2} \ldots w_k \star^{n_k}$$
where for all $1\leq i \leq k$, $w_i \in \alphabet^+$ and for all 
$0< i < k$, $n_i > 0$, $n_0\geq 0$ and $n_k \geq 0$.
We get: 
$$\fonc{\A'}(w)=\norm[\infty]{\mu'(w)}\ge \sum_{i=1}^k I\mu(w_i)F$$
since the weight of $\star$ is $0$. This is at least $\sum_{i=1}^k \length{w_i}\inf_u\frac{\fonc\A(u)}{\length u}$.
If $\fonc\A\ge0$, then we get $\fonc{\A'}(w)\ge0$.
Otherwise, $\inf_u\frac{\fonc\A(u)}{\length u}<0$, and since $\sum_i\length{w_i}\le\length w$, we get 
$\frac{\fonc{\A'}(w)}{\length w}\ge\inf_u\frac{\fonc\A(u)}{\length u}$.
Moreover, since the weights of the words in $\{\star\}^*$ is $0$ in $\A'$, then the inequality holds.

The other inequality is obtained by observing how arcs labelled $\star$ are positioned in $\A'$. Indeed, if a transition labelled by $\star$ is taken, then it has to start from a final state or $q$, and has to end in an initial state or $q$. Moreover, no other letter labels a transition starting or ending in $q$. So, when reading a word $w \in \alphabet^+$ between two $\star$, this word is read on a path that was already an existing accepting path in $\A$.

Thus, we see that for all words $w\in\alphabet^+$
and all $k\in\N$: 
$$\fonc{\A'}((\star w)^k\star)=k\fonc\A(w)$$ 
so that:
\begin{align*}
\inf_{u\in{\alphabet'}^+}\frac{\fonc{\A'}(u)}{\length u} & \le\inf_{\substack{w\in\alphabet^+ \\ k\in\N}}\frac{\fonc{\A'}((\star w)^k\star)}{k(\length w+1)+1} \\
& \le\inf_{\substack{w\in\alphabet^+ \\ k\in\N}}\frac{\fonc{\A}(w)}{\length w+1+1/k}
\end{align*}
\end{proof}

As a corollary of this lemma and of the previous results on the 
undecidability of comparison, we get the following theorem.
\begin{theorem}
\label{theorem:if}
The restrictions of Problems $\Comp3{\{-1,0,1\}}$ and $\Comp[554]7\Zm$ to automata whose states are all initial and final are still
 undecidable.
\end{theorem}


%
%

\section{Uncomputability of the joint spectral radius of tropical matrices}
\label{section:jsr}

\subsection{Joint spectral radius} 

The definition of spectral radius extends to the \dfn{joint spectral radius} of a set
$\Gamma\subseteq \matrices{d}{\Zm}$ of matrices, as follows:
\[
\rho(\Gamma)=\inf_{\ell > 0}
\sett{\frac1\ell \norm[\infty]{M_1\dotsm M_\ell}}{M_1, \ldots, M_\ell 
\in \Gamma}
\]

The following lemma, which gives other equivalent definitions\footnote{
Note that here we use the $\inf$ definition for the joint spectral radius instead of 
the $\sup$ definition used in the literature. The latter is easy to compute in $\Zm$, 
unlike the notion considered here (sometimes called lower spectral radius of joint spectral subradius).},
is a known application of Fekete's subadditive
lemma (see for example \cite[Theorem 3.4]{GaubertMairesse98}). 

\begin{lemma}
\label{l:defrho} 
For any set $\Gamma$ of matrices in $\matrices{d}{\Zm}$,
we have:
\begin{align*}
&\rho(\Gamma)\\
&=\lim_{\ell\to\infty}\min\sett{\frac1\ell\norm[\infty]{M_1\cdots M_\ell}}{M_1,\ldots,M_\ell\in\Gamma}\\
&=\inf_{\ell > 0}
\sett{\frac1\ell{\rho(M_1\dotsm M_\ell)}}{M_1, \ldots, M_\ell \in \Gamma}
\end{align*}
\end{lemma}
\begin{proof}
Let: 
$$u_\ell = \inf
\sett{\norm[\infty]{M_1\dotsm M_\ell}}{M_1, \ldots, M_\ell \in \Gamma}$$
The sequence $(u_\ell)_\ell$ is subadditive \ie for all $\ell,\ell'$, 
$u_{\ell+\ell'} \leq u_\ell + u_{\ell'}$. Indeed for all 
$M_1, \ldots, M_{\ell+\ell'} \in \Gamma$, 
\begin{align*}
&\norm[\infty]{M_1\dotsm M_{\ell+\ell'}} \\
&\leq \norm[\infty]{M_1\dotsm M_\ell} + 
\norm[\infty]{M_{\ell+1}\dotsm M_{\ell+\ell'}}
\end{align*}
Thus by Fekete's lemma, 
$\lim_{\ell\to\infty} \frac{u_\ell}{\ell}$ is well defined and
$ \inf_{\ell > 0}\frac{u_\ell}{\ell} = \lim_{\ell\to\infty} \frac{u_\ell}{\ell}$, 
which implies
\begin{align*}
&\rho(\Gamma) \\
&= \lim_{\ell\to\infty} \inf
\sett{\frac1\ell\norm[\infty]{M_1\dotsm M_\ell}}{M_1, \ldots, M_\ell \in \Gamma}
\end{align*}

As for the second equality, let us denote: 
\[\rho'(\Gamma)= \inf_{\ell > 0}
\sett{\frac1\ell{\rho(M_1\dotsm M_\ell)}}{M_1, \ldots, M_\ell \in \Gamma}\]
Since for all matrices $M$,
$\rho(M)\leq \norm[\infty]{M}$, we have $\rho'(\Gamma) \leq \rho(\Gamma)$.
Let us show the reverse inequality.
For all $\varepsilon>0$, there is $\ell>0$ and $M_1, \ldots, M_\ell \in \Gamma$ 
such that $\frac1\ell \rho(M_1\cdots M_\ell) \leq \rho'(\Gamma) + \varepsilon$.
By definition, it means that $\frac1\ell \lim_{n} 
\frac{1}{n} \norm[\infty]{(M_1\cdots M_\ell)^n} \leq \rho'(\Gamma) + \varepsilon$,
or equivalently, $\lim_{n}  \frac{1}{n \ell} \norm[\infty]{(M_1\cdots M_\ell)^n} 
\leq \rho'(\Gamma) + \varepsilon$.
By definition, $\rho(\Gamma) \leq \lim_{n}  \frac{1}{n \ell} 
\norm[\infty]{(M_1\cdots M_\ell)^n}$, thus, for all $\varepsilon > 0$, 
$\rho(\Gamma) \leq \rho'(\Gamma) + \varepsilon$, that concludes the proof.
\end{proof}

It can be easily seen that $\rho(k\odot\Gamma)=\rho(\Gamma)+k$.

\subsection{Ultimate rank}

In the classical setting of a field, the notion of rank enjoys many equivalent definitions.
These notions do not coincide in the case of $\Zm$.
However, it was noticed in~\cite{urk} that they coincide on the limit points of 
the powers of the matrix, when properly normalized (or considered projectively).
This is formalized in \cite[Theorem 5.2]{urk}, and equivalent to the following definition:
the \dfn{ultimate rank} $\urk(M)$ of a matrix $M \in \matrices{d}{\Zm}$ is the
sum of the cyclicities of the strongly connected components of its critical graph.
Clearly, $\urk(M)=0$ ($M$ has empty critical graph) if and only if $\rho(M)=-\infty$,
and this corresponds to the nilpotency of~$M$.

As for the joint spectral radius, this notion can be generalized to sets of matrices.
The \dfn{ultimate rank} of a set $\Gamma \subseteq \matrices{d}{\Zm}$ of matrices
is: 
$$\urk(\Gamma)=\min\sett{\urk(M)}{M\in\semi\Gamma}$$
Clearly, $\urk(\Gamma)=0$ if and only if $\rho(\Gamma)=-\infty$,
and this corresponds to the mortality of the semigroup generated by~$\Gamma$.
It can be seen (or read in \cite[Theorem 5.2]{urk}) that the ultimate rank is a projective notion:
$\urk(k\odot\Gamma)=\urk(\Gamma)$ for any~$k\in\Z$.

{
In some interesting cases, $\urk(\Gamma)$ is indeed the reached minimum of the ranks in the semi-group
, so that it is the dimension of the
limit set of the action of~$\Gamma$ on~$\R^d$.
Those cases include sets with irreducible fixed structure (all matrices have the same infinite entries),
and sets of matrices with no line of~$-\infty$ that contain one matrix with only finite entries.
This is implicitely used in~\cite{MERLET2010} and allows to extend some nice properties of products of random matrices
from matrices with the so-called memory-loss property (case~$\urk(\Gamma)=1)$ to more general ones (\cite[Corollary 1.2]{MERLET2010}).
}
\subsection{Uncomputability and link with automata}

Finitely generated semigroups of matrices 
exactly correspond to max-plus automata that 
have all their states initial and final. 
In particular, the following lemma links the computation of the joint spectral radius for the former to the comparison for the latter.

\begin{lemma}
\label{l:negrad}
Let $\A$ be a max-plus automaton over an alphabet
$\alphabet$ whose all states are both initial and final.
The following statements are equivalent.
\begin{enumerate}
\item $\fonc{\A} \geq 0$. \label{l:negrad-un}
\item For all matrices $M$ in $\semi{\Gamma_\A}$, $\norm[\infty]{M} \geq 0$. 
\label{l:negrad-deux}
\item For all matrices $M$ in $\semi{\Gamma_\A}$, $\rho(M) \geq 0$.
\label{l:negrad-quatre}
\item $\rho(\Gamma_{\A}) \geq 0$.
\label{l:negrad-cinq}
\end{enumerate}
\end{lemma}
According to the terminology in \cite{Almagor2011}, this also corresponds to the case when $\A$ is called \dfn{universal with threshold $0$}.
\begin{proof}
Items~\ref{l:negrad-un}. and~\ref{l:negrad-deux}. are equivalent since all the
states of $\A$ are both initial and final. Thus, 
for all words $w$, $\fonc{\A}(w) = \norm[\infty]{\mu(w)}$. Moreover,
$\semi{\Gamma_\A}$ is exactly the set $\sett{\mu(w) }{ w\in \alphabet^+}$.

Items~\ref{l:negrad-deux}. and~\ref{l:negrad-quatre}. are equivalent by definition
of the joint spectral radius.

Finally, Items~\ref{l:negrad-quatre}. and~\ref{l:negrad-cinq}. are equivalent by 
Lemma~\ref{l:defrho}. 
\end{proof}

The uncomputability of the joint spectral radius is deduced from the equivalence in Lemma~\ref{l:negrad} and from Lemma~\ref{lemma:if}.
More precisely, define $\JSR kS$ \resp{$\JSR[d]kS$} as the following problem: 
\textit{Given a finite set of $k$ matrices with coefficients in $S \subseteq \Zm$ 
\resp{and dimension $d$},
determine whether their joint spectral radius
is greater than or equal to $0$.}

\begin{theorem}
\label{t:fixedw}
Problems $\JSR3{\{-\infty,-1,0,1\}}$ and $\JSR[554]7\Zm$ are undecidable.
\end{theorem}
\begin{proof}
The undecidability comes from a reduction from the problem stated in Theorem~\ref{theorem:if}.
Consider a max-plus automaton $\A$ whose states are all initial and final. 
By Lemma~\ref{l:negrad}, $\fonc{\A} \geq 0$ if and only if the joint spectral
radius of $\Gamma_{\A}$ is nonnegative. Thus $\JSR3{\{-\infty,-1,0,1\}}$
and $\JSR[554]7\Zm$ are undecidable.
\end{proof}



By reduction from Theorem~\ref{t:fixedw}, 
we prove that the ultimate rank is also uncomputable.
Define $\UR kS$ \resp{$\UR[d]kS$} as the following problem: 
\textit{Given a finite set of $k$ matrices with coefficients in $S$ 
\resp{and dimension $d$},
determine whether the ultimate rank of the semigroup that they generate 
is equal to $1$.}

\begin{theorem}
\label{t:ur}
Problems $\UR3{\{-\infty,-1,0,1\}}$ and $\UR[1109]7\Zm$ are undecidable.
\end{theorem}


\begin{proof}
From any matrix $M$, one can build:
\[\widehat M=\left[\begin{array}{ccc}M&\z&\z\\\z&M&\z\\\z&\z&0
\end{array}\right]~.\]
It is then clear that, for any finite family of matrices $\Gamma$, 
the semigroup generated by $\widehat\Gamma=\sett{\widehat M}{M\in\Gamma}$ is 
$\semi{\widehat\Gamma}=\sett{\widehat M}{M\in\semi\Gamma}$.

If $M$ has size $d$ and entries in $S$, then $\widehat M$ has size $2d+1$ and entries in $S\cup\{-\infty,0\}$. 
Moreover, if $\rho(M)<0$, then the critical graph of $\widehat M$ 
is simply the loop over the last vertex (last line of the matrix 
$\widehat M$), so that $\urk(\widehat M)=1$.
Otherwise, the critical graph of $\widehat M$ contains at least two copies 
of that of $M$ (which is nonempty), so that $\urk(\widehat M)\ge2$.
Thus, $\rho(M)\geq 0$ if and only if $\urk(\widehat M)\ge2$.
By reduction from the undecidable problems of Theorem~\ref{t:fixedw},
we can deduce that
$\UR3{\{-\infty,-1,0,1\}}$ and $\UR[1109]7\Zm$
are undecidable.
\end{proof}

\begin{remark}
As noted above, the joint spectral radius and ultimate rank are not altered through translation by a constant; 
thus uncomputability is preserved with other restrictions over the entries.
Regarding the joint spectral radius, the comparison
to $0$ may no longer be undecidable, but the comparison to some other constants
remains undecidable like, for example,
if $\Gamma\subset\M(\Nm)$, whether $\rho(\Gamma) \geq 1$.
\end{remark}

\subsection{Approximation of the joint spectral radius}

Still by using results in automata theory, 
we prove that even though the joint spectral radius is not computable 
in general, it is approximable 
and computable in restricted cases
in the following sense.

\begin{theorem}
\label{theorem:approx}
There is an algorithm that, given a finite set $\Gamma$ of matrices
 and $n\in\Ns$, computes a value $\alpha \in \Q \cup \{-\infty\}$ such that 
$\alpha-\frac1n\leq \rho(\Gamma) \leq \alpha+\frac1n$. 
\end{theorem}

The proof uses the main result of \cite{ColDav13}. This result is originally
stated for min-plus automata using only positive weights. These automata are defined over
the \dfn{min-plus semiring} $(\Z \cup \{+\infty\}, \min, +, +\infty, 0)$. By using 
the morphism from the min-plus to the max-plus semiring that associates 
$k$ to $-k$, we can state the result of \cite{ColDav13} in the max-plus case.

\begin{proof}
First, let us exhibit an algorithm $\mathfrak A$ that gives an approximation of the joint spectral radius 
of any finite set of matrices 
with only nonpositive entries.
Consider a finite set of matrices $\Gamma$
 with only nonpositive entries, and a
max-plus automaton $\A$ such that $\Gamma=\Gamma_\A$. 
From \cite{ColDav13}, there is an algorithm $\mathfrak A$ that, given a max-plus automaton $\A$ over an alphabet
$\alphabet$ using only nonpositive weights
and $n\in\Ns$, computes a value $\alpha \in \Q \cup \{+\infty\}$ such that:
$\alpha-\frac1n \leq \inf_{w\in \alphabet^+}\frac{\fonc{\A}(w)}{|w|}
\leq \alpha + \frac1n$.
$\mathfrak A$ also gives an approximation of the joint spectral radius
of $\Gamma$, since: 
\begin{align*}
\rho(\Gamma)&=&&\inf_{\ell > 0}
\sett{\frac1\ell \norm[\infty]{M_1\dotsm M_\ell}}{M_1, \ldots, M_\ell 
\in \Gamma} \\
&= && \inf_{\ell > 0}
\sett{\frac1\ell \fonc{\A}(w)}{w\in \alphabet^\ell} \\
&= && \inf_{w\in \alphabet^+}\frac{\fonc{\A}(w)}{|w|}
\end{align*}

Consider now a finite set of matrices $\Gamma$ with arbitrary entries.
Let $k$ denote the greatest entry that appears in at least one of the matrices 
of $\Gamma$.
Construct the set $\Gamma'=-k\odot\Gamma$. The set $\Gamma'$ is then a finite set of matrices 
with only nonpositive
entries, on which we can apply $\mathfrak A$. We then get an 
approximation of the joint spectral radius of $\Gamma$ by adding $k$ to the 
value given by the algorithm.
\end{proof}
This implies, in particular, that the joint spectral radius of every finite set of matrices is a computable real number.

\paragraph*{Remarks about the complexity}
The algorithm of~\cite{ColDav13} is EXPSPACE in the size of the automaton 
and in $n$. 
Moreover the problem is PSPACE-hard by reduction from the universality problem 
of a nondeterministic automaton:
\textit{Given a nondeterministic finite automaton $\A$ over 
a $2$-letter alphabet $\alphabet$, the problem to determine whether 
the language accepted by $\A$ is $\alphabet^+$ is PSPACE-complete.}
A precise statement of the reduction is given in the following lemma.

\begin{lemma}
\label{lemma:pspace-hard}
Given a nondeterministic finite automaton $\A$ over a $2$-letter alphabet $\alphabet$,
one can construct in polynomial time a set of $3$ matrices $\Gamma$ with entries in 
$\{-\infty, 0\}$ such that $\A$ accepts $\alphabet^+$ if and only if the joint spectral
radius of $\Gamma$ is equal to $0$. Otherwise, the joint spectral radius of 
$\Gamma$ is equal to $-\infty$.
\end{lemma}

\begin{proof}
Consider a nondeterministic finite automaton $\A$ over 
a $2$-letter alphabet $\alphabet$. We construct a max-plus automaton $\A'$
from $\A$ by weighting the transitions by $0$. Then, $\A$ accepts $\alphabet^+$ if and only if
$\fonc{\A'} = 0$ (otherwise there is a word $w$ such that 
$\fonc{\A'}(w) = -\infty$).
By Lemma~\ref{lemma:if}, one can construct a max-plus automaton $\B$ over a $3$-letter alphabet such that 
every state of $\B$ is both initial and final, $\B$ has only weight $0$ on
its transitions, and
$\fonc{\A'} \geq 0$ if and only if $\fonc{\B} \geq 0$. 
Hence, $\fonc{\A'} = 0$ if and only if $\fonc{\B} = 0$.
By Lemma~\ref{l:negrad}, $\fonc{\B} \geq 0$ if and only if the joint spectral radius
of $\Gamma_\B$ is nonnegative. Since, $\Gamma_\B$ contains only matrices with entries
in $\{0,-\infty\}$, it implies that $\fonc{\B} = 0$ if and only if the joint 
spectral radius of $\Gamma_\B$ is equal to $0$. 
All the constructions are polynomial.
\end{proof}

Notice that Lemma~\ref{lemma:pspace-hard} also proves that $\JSR3{\{0,-\infty\}}$
is PSPACE-hard.
A result in~\cite{Almagor2011} implies that $\JSR k{\Z^{-} \cup \{-\infty\}}$ is also
PSPACE, where $\Z^{-}$ denotes the set of nonpositive integers. Hence, 
Problem $\JSR3{\{0,-\infty\}}$ is also PSPACE-complete.

\subsection{Restriction to finite entries}

Let us consider the restriction to matrices with only finite entries.
In terms of automata, it means that for all letters $a$, 
there is a transition labelled by $a$
between any pair of states. In this case, the joint spectral radius and 
ultimate rank are computable.

\begin{theorem}
\label{t:FiniteEntries}
There are PSPACE algorithms to compute the joint spectral radius and 
the ultimate rank of any finite set of matrices with finite entries. In particular,
in this case, the joint spectral radius is a rational number.
Moreover, $\JSR3{\{0,-1\}}$ and $\UR3{\{0,-1\}}$ are PSPACE-complete.
\end{theorem}

The fact that the problems are PSPACE-hard comes from Lemma~\ref{lemma:pspace-hard},
that proves that $\JSR3{\{0,-\infty\}}$ is PSPACE-hard and by
successive reductions from $\JSR3{\{0,-\infty\}}$ to $\JSR3{\{0,-1\}}$
and from $\JSR3{\{0,-1\}}$ to $\UR3{\{0,-1\}}$. 

As for proving that the problems are PSPACE,
the key of the reasoning is the following lemma:

\begin{lemma}[\cite{Gaubert96Burnside}]
\label{lemma:gaubert}
Let $\Gamma\subset\matrices{d}{\{-b,\ldots,b\}}$ for some nonnegative integers $b$ and $d$.
Then for all matrices $M\in\semi\Gamma$ and all indices $i,j$,
the quantity $M_{i,j}-M_{1,1}$ belongs to $\{-2b,\ldots,2b\}$.
\end{lemma}

\begin{proof}[Proof of Theorem \ref{t:FiniteEntries}]
{
Consider a finite set $\Gamma$ of matrices which have only entries
in $\{-b,\ldots,b\}$. 
By Lemma~\ref{lemma:gaubert}, the set $\Lambda = \sett{-M_{1,1}\odot M}{M \in \semi{\Gamma}}$ contains
at most $(4b+1)^{d^2-1}$ matrices.

Moreover, since the operation of adding the same constant to all the entries of a matrix
commutes with the product of matrices, $\Lambda$ is 
the set of matrices $-M_{1,1}\odot M$ such that $M$ is a product of at most 
$(4b+1)^{d^2-1}$ matrices of $\Gamma$.
Finally, 
the ultimate rank of $\Gamma$ is minimum of
the ultimate rank of the matrices in $\Lambda$, which can be computed by the following algorithm in NPSPACE.
Start with a matrix~$M=M_1\in\Gamma$ and a counter~$\ell$ with value~$1$.
At each (nondeterministic) step, either compute~$\urk(M)$ and stop, or increase~$\ell$ by one and multiply~$M$ by some matrix of~$M_\ell\in\Gamma$.
If $\ell=(4b+1)^{d^2-1}$, then compute~$\urk(M)$ and stop. 

Since the maximum value of~$\ell$ is simply exponential in the size
$|\Gamma|d^2\log(b)$ of the input, both $\ell$ and the size of the entries of $M=M_1\cdots M_\ell$
are simply exponential and thus can be stored in polynomial space.
Since the product of matrices and ultimate rank of one matrix can be computed in~P
the algorithm is in~NPSPACE=PSPACE.

Concerning the joint spectral radius, let us prove that 
\begin{equation}\label{e:bndedPdts}
\rho(\Gamma)=\min_{\substack{\ell \leq (4b+1)^d \\ M_1,\ldots, M_\ell \in \Gamma}} \left\{\frac1\ell \rho(M_1\cdots M_\ell)\right\}
\end{equation}
and conclude in the same way.
To prove~\eqref{e:bndedPdts}, consider a product $M_1\cdots M_\ell$ of matrices in~$\Gamma$ 
and the orbit of the vector with all entries equal to~$0$
under the action of~$M_1, M_2, \ldots, M_\ell, M_1, M_2, \ldots$. 
By Lemma~\ref{lemma:gaubert},
this orbit projectively has size at most $(4b+1)^{d}$.
Hence, it cycles after $t$ steps for some $t \leq (4b+1)^{d}$ 
and has a period $p \leq (4b+1)^{d}$.
Each time the orbit goes back to the same vector projectively,
all coordinates have increased by some value,
which is the spectral radius of 
$M_{(t+1)\bmod \ell} M_{(t+2) \bmod \ell} \cdots M_{(t+p) \bmod \ell}$. 
Indeed, 
for matrices with only finite entries, the spectral radius is
the only eigenvalue.
Finally, we get 
$\frac1\ell\rho(M_1\cdots M_\ell)=\frac{1}{p}\rho(M_{(t+1)\bmod \ell} 
M_{(t+2) \bmod \ell} \cdots M_{(t+p) \bmod \ell})$.
}
\paragraph*{PSPACE-hardness}
Let $\Gamma$ be a finite set of matrices with entries in $\{0,-\infty\}$. Let $\Gamma'$
be the set $\Gamma$ where every entry with value $-\infty$ has been replaced by $-1$.
The joint spectral radius of $\Gamma$ is equal to $0$ if and only if the joint spectral
radius of $\Gamma'$ is equal to $0$. Otherwise, $\rho(\Gamma) = -\infty$ and 
$\rho(\Gamma')$ is strictly negative. Thus, $\JSR3{\{0,-1\}}$ is PSPACE-hard.

Now, let us reduce $\JSR3{[\{0,-1\}}$ to $\UR3{\{0,-1\}}$.
From any matrix $M\in\matrices d{\{0,-1\}}$, with $d\in\Ns$, one can build the matrix:
\[\widetilde M=\left[\begin{array}{cc}M&(-1)\\(-1)&0
\end{array}\right]\in~\matrices{d+1}{\{0,-1\}}~,\]
where $(-1)$ is the vector with appropriate size whose entries are all~$-1$.

It is then clear that, for any finite family of matrices $\Gamma$, the semigroup generated by $\widetilde\Gamma=\sett{\widetilde M}{M\in\Gamma}$ is $\semi{\widetilde\Gamma}=\sett{\widetilde M}{M\in\semi\Gamma}$.

Moreover, note that if $\rho(M)<0$, then the critical graph of $\widetilde M$ is the loop over the last vertex, so that $\urk(\widetilde M)=1$.
Otherwise, $\rho(M)=0$, and the critical graph is the union of this loop and the critical graph of $M$, so that $\urk(\widetilde{M})=1+\urk(M)$.
We deduce that the ultimate rank of $\widetilde{\Gamma}$ is greater than or equal to $2$
if and only if $\rho(\Gamma) \geq 0$, which we already know to be PSPACE-hard.
\end{proof}

{
\begin{remark}
Lemma~\ref{lemma:gaubert} is implicitely used in~\cite[Corollary~2]{Gaub95} to prove that the functions computed by max-plus automata with rational entries
whose linear representation generates a so-called primitive semigroup (which includes matrices with finite entries)
can be computed by a deterministic automaton.

It is also shown (as Corollary~4) that the minimal growth rate of a deterministic automaton,
\ie the joint spectral radius of its linear representation, can be computed
as the spectral radius of one matrix whose indices are the states of the automaton.

This gives another algorithm to compute the joint spectral radius of a finite set of matrices with finite integers, but not a PSPACE one, since
the size of the matrix is only bounded by~$(4b+1)^{d^2}$, while Equation~\eqref{e:bndedPdts} allows to compute only matrices of size~$d$
without storing them.
\end{remark}
}

%
%

\section{Encoding of two-counter machines into weighted automata}
\label{a:comparaison}

In this section, we give the complete proof of Theorem~\ref{theorem:main}.

\paragraph*{Two-counter machines}

Several variants of two-counter machines exist, all equivalent in terms of 
expressiveness. We use here the one described in \cite{Ivanov2014}:
A \dfn{two-counter machine} is a deterministic finite state machine with 
two counters that can be incremented or decremented if not valued to $0$. 
More precisely, it is given by a tuple 
$(Q,T^+_1,T^+_2,T^-_1,T^-_2,q_{init},q_{halt})$
where $Q$ is a finite set of states, $T^+_1$\resp{$T^+_2$} is a finite set
in $Q^2$ which represents the transitions that increment the 
first\resp{second} counter, $T^-_1$\resp{$T^-_2$} is a finite set
in $Q^3$ which represents the transitions that check if the 
first\resp{second} counter is valued to $0$ and decrement it if not,
$q_{init} \in Q$ is the initial state and 
$q_{halt} \in Q$ is a final state
such that there is no outgoing transition from $q_{halt}$ 
(for all transitions $(q,p) \in T_1^+ \cup T_2^+$ or 
$(q, p, p') \in T^-_1\cup T^-_2$,
$q\neq q_{halt}$).
Moreover the machine is deterministic:
in one state there is at most one action that can be performed,
\ie for all $q\in Q$, there is at most one transition of the form 
$(q,p)$ or $(q,p,r)$ in $T^+_1 \cup T^+_2 \cup T^-_1 \cup T^-_2$ 
and $T^+_1 \cap T^+_2 = \emptyset$ and $T^-_1 \cap T^-_2 = \emptyset$.

The semantics of a two-counter machine is given by means of the 
valuations of the counters
that are pairs of nonnegative integers. 
An execution with counters initialised to $(n_1^{0},n_2^{0})$
is a sequence of transitions and valuations denoted by:
$$(n_1^{0},n_2^{0}) \xrightarrow{t_1} (n_1^{1},n_2^{1}) 
\xrightarrow{t_2} (n_1^{2},n_2^{2}) \ldots 
\xrightarrow{t_k} (n_1^{k},n_2^{k})$$
such that: 
\begin{itemize}
\item for all $i \in \{1,\ldots, k\}$, 
if $t_i \in T^+_1$\resp{$T^+_2$} then $n_1^{i} = n_1^{i-1}+1$ and 
$n_2^{i} = n_2^{i-1}$\resp{$n_1^{i} = n_1^{i-1}$ and 
$n_2^{i} = n_2^{i-1}+1$},
\item for all $i \in \{1,\ldots, k\}$, 
if $t_i \in T^-_1$\resp{$T^-_2$} then $n_1^{i}=n_1^{i-1}=0$
or $n_1^{i} = n_1^{i-1}-1$ and 
$n_2^{i} = n_2^{i-1}$\resp{$n_1^{i}=n_1^{i-1}$
and $n_2^{i} = n_2^{i-1}=0$ or $n_2^{i} = n_2^{i-1}-1$},
\item for all $i \in \{1,\ldots, k-1\}$, 
if $t_i=(p_i,q_i)\in T^+_1\cup T^+_2$ then 
$t_{i+1} \in \{q_i\}\times (Q \cup Q^2)$,
\item for all $i \in \{1,\ldots, k-1\}$, 
if $t_i=(p_i,q_i,q'_i)\in T^-_1\cup T^-_2$ and $n_{i-1}=0$ then 
$t_{i+1} \in \{q_i\} \times (Q\cup Q^2)$, otherwise if $n_{i-1} \neq 0$ then 
$t_{i+1} \in \{q'_i\} \times (Q\cup Q^2)$.
\end{itemize}

The machine halts with counters initialised to $(n,m)$ 
if and only if the execution with counters initialised to $(n,m)$ ends in 
$q_{halt}$.

The halting problem for two-counter machines when counters are initialised to $(0,0)$ 
is undecidable. Moreover, like for Turing machines, there exists a, specific, so-called \dfn{universal} two-counter machine (U2CM), that is able 
to simulate the behaviour of any two-counter machine: there is an encoding $\delta$ from the 
two-counter machines to the positive integers such that the U2CM
halts with counters
initialised to $(\delta(\M),0)$ if and only if $\M$ halts with counters
initialised to $(0,0)$.
Therefore, 
 the problem to determine, given a
positive integer $n$, whether this particular U2CM halts
when counters are initialised to $(n,0)$, is undecidable.

In \cite{Ivanov2014}, it is proved that there exists such a machine with $268$ states. As far as we know, it is the best known bound. 


We prove the following result:

\begin{lemma}
\label{lemma:transfoweights}
Given a two-counter machine $\mathcal{M}$ with $d$ states 
and a non-negative integer $n$, 
one can build a max-plus automaton $\A$ on a $6$-letter alphabet
with weights in $\{-n-1,-2,-1,0,1,n-1\}$ and $2d+27$ states such that 
$\mathcal{M}$ halts with counters initialised to $(n,0)$
if and only if there is a word $w$ such that $\fonc{\A}(w) < 0$.
\end{lemma}

The main difference with the previous proof of~\cite{Colcombet} comes from the fact that we encode a universal two-counter machine starting with any integer $n$ in the first counter, 
when the previous ones encode any two-counter machine starting with null counters. 
We thus need to ensure that the first counter starts with n (point 4 in the proof of Lemma~\ref{lemma:transfoweights}). 
So, instead of max-plus automata with bounded weights and an unbounded number of states, we obtain max-plus automata with a bounded number of states, but with unbounded weights.

\begin{proof}
Let $\Omega = \{c^+_1,c^+_2,c^-_1,c^-_2\}$ denote the set of the possible actions on the counters. The idea is to encode the executions in $\mathcal{M}$ by words on the alphabet
$\alphabet = \{a,b\} \cup \Omega$. 
A block $a^m$\resp{$b^m$} will encode the fact that the 
first\resp{second} counter has currently value $m$. 
For example, a word $a^nb^m c_1^{+} a^{n+1}b^m c_2^{-} a^{n+1}b^{m'}$ 
encodes an execution starting with value $n$ in the first counter and $m$ 
in the second counter.
The action $c_1^{+}$ is performed, leading to store $n+1$ in the first 
counter and $m$ in the second, encoded by the word $a^{n+1}b^m$. Then the action
$c_2^{-}$ is used. In this case, either $m=m'=0$ or $m=m'+1$.

The max-plus automaton $\A$ is constructed in such a way that if
a word $w$ encodes a valid execution from state $q_{init}$ to state $q_{halt}$, 
starting with value $n$ in the first counter and value $0$ in the second, 
then $\fonc{\A}(w) = -1$, otherwise 
$\fonc{\A}(w) \geq 0$. 

The automaton $\A$ is constructed as the finite union of automata:
some checks that $w$ has the good shape \ie it belongs to $a^*(\Omega a^*b^*)^*\Omega$ and represents a valid path with respect to the states of $M$.
In this case the value associated with $w$ is $-\infty$, otherwise it is $0$.
Another automaton checks that the counters are correctly incremented/decremented.
Finally, the main difference with the proof by Colcombet, is to check that the counters are initialised to $(n,0)$ (and not $(0,0)$), \ie check that 
the word belongs to $a^n \Omega \alphabet^*$.

More precisely, we consider a two-counter machine 
$\mathcal{M} = (Q,T^+_1,T^+_2,T^-_1,T^-_2,q_{init},q_{halt})$
and a nonnegative integer $n$. A word is not valid 
(does not represent an accepting run in $\mathcal{M}$) if it satisfies at least 
one of the following conditions:
\begin{enumerate}
\item it does not belong to $a^*\Omega(a^*b^*\Omega)^*\Omega$,
\label{forme}
\item it does not correspond to a path in machine $\M$ (with respect to
the states of $\M$), \label{match}
\item the counters are badly incremented/decremented. 
For example, a word contains a factor of the form 
$\gamma a^nb^m c_1^{+} a^{n'}b^{m'} \gamma'$ 
with $n'\neq n+1$ or $m\neq m'$, \label{incdec}
\item it does not belong to $a^n \Omega \alphabet^*$ 
(the first counter is not initialised to $n$). \label{an}
\end{enumerate}
Items \ref{forme}. and \ref{match}. refer to rational languages. It is then sufficient 
to construct an automaton $\A_1$ recognising these languages. 
By weighting all the transitions
by $0$, a word $w$ that satisfies at least one of the conditions \ref{forme}. 
or \ref{match}. verifies $\fonc{\A_1}(w)=0$. Otherwise, $\fonc{\A_1}(w)=-\infty$.

The automaton corresponding to the condition \ref{forme}. 
is given in Figure~\ref{figure1}.

\begin{figure}[!htbp]
\begin{center}
\begin{tikzpicture}[scale=0.5]
\tikzset{
	every state/.style={scale=0.7,draw=blue!50,very thick,fill=blue!20},
	initial text=,
	accepting/.style={accepting by arrow},
	fleche/.style={->, >=latex}
}

\node[state, initial] (q_1) at (2,0) {};
\node[state, accepting] (q_2) at (6,0) {$q$};
\node[state, initial] (q_3) at (-6,0) {$p$};
\node[state, accepting] (q_4) at (-2,0) {};
\node[state] (q_5) at (0,-3) {};

\path[fleche]     (q_1) edge  node [above] {\small{$b:0$}}
 (q_2);
\path[fleche]     (q_2) edge [loop above] node         {\small{$\alphabet:0$}} ();
\path[fleche]     (q_3) edge [loop above] node         {\small{$\alphabet:0$}} ();
\path[fleche]     (q_3) edge  node [above] {\small{$a,b:0$}} 
(q_4);
\path[fleche]     (q_3) edge  node [above, sloped] {\small{$b:0$}} (q_5);
\path[fleche]     (q_5) edge  node [above, sloped] {\small{$a:0$}} (q_2);
\end{tikzpicture}
\end{center}
\caption{\label{figure1}Item \ref{forme}.}
\end{figure}

The automaton constructed for Item \ref{match}. has to check if the execution is not valid with respect to the states of $\M$. 
Essentially, it follows the execution, goes in a sink state if a transition does not exist and accepts all the words that do not end in $q_{halt}$.
Let us remind that $\M$ is deterministic. Let us split $Q$ into the disjoint union of $Q^+ \cup Q_1^- \cup Q_2^-$ where
$Q^+$ represents the states where the actions performed are only increment, 
$Q^-_1$ the states where the first counter can be decremented and 
$Q^-_2$ the states where the second counter can be decremented.
The automaton constructed has states $Q \cup \{q^a \mid q \in Q^-_1\}  \cup \{q^b \mid q \in Q^-_2\} $, plus an additional state to make the automaton complete. 
The initial state is $q_{init}$. 
If we are in a state $q$ of $Q^-_1$\resp{$Q^-_2$}  and we read an $a$\resp{$b$}, we move to $q^a$\resp{$q^b$}, \ie there are transitions 
$(q,a,q^a)$ for all $q\in Q^-_1$\resp{$(q,b,q^b)$ for all $q\in Q^-_2$}. 
Moreover there are loops $(q,b,q)$, $(q^a,a,q^a)$ and $(q^a,b,q^a)$ for all $q\in Q^-_1$\resp{$(q,a,q)$, $(q^b,a,q^b)$ and $(q^b,b,q^b)$ for all $q\in Q^-_2$}
as well as loops $(q,a,q)$ and $(q,b,q)$ for all $q\in Q^+$.

Finally, there are transitions $(p,c_1^+,q)$\resp{$(p,c_2^+,q)$}
for all $(p, q) \in T^+_1$\resp{$T^+_2$},
transitions $(p,c_1^-,q)$\resp{$(p,c_2^-,q)$} if for some $r$, 
$(p, q, r) \in T^-_1$\resp{$T^-_2$}  and 
transitions 
$(p^a,c_1^-,r)$\resp{$(p^b,c_2^-,r)$} if for some $q$, 
$(p, q, r) \in T^-_1$\resp{$T^-_2$}.

All the transitions have weight $0$. All the states are final except $q_{halt}$.

This automaton has at most
$2|Q| + 1$ states (to minimise the global number of states, the additional state
can be merged with the state $q$ of Figure~\ref{figure1}).


As for item \ref{incdec}., let us treat the case of checking if the 
first counter is well incremented after performing a transition with action
$c_1^{+}$.
We construct an automaton $\A_2$ 
as in Figure~\ref{figure:krob2}. 
Consider the word
$w = a^nb^m c_1^{+} a^{n'}b^{m'}$.
For the first part of $\A_2$ (above), the value computed on $w$ is $n-n'$. 
As for the second part (below), the value computed on $w$ is $-n-2+n'$. 
Thus, $\fonc{\A_2}(w) = \max(n-n',n'-n-2)$. 
If $n' \neq n+1$, 
$\fonc{\A_2}(w) \geq 0$, otherwise $\fonc{\A_2}(w) = -1$.
By nondeterminism and the use of the semantics $\max$, we can prove that
for all the words $w$,
$\fonc{\A_2}(w) \geq 0$ if and only if $w$ contains a factor witnessing that
the first counter is badly incremented after performing a transition with an 
action $c_1^{+}$. Otherwise, $\fonc{\A_2}(w) = -1$.
To check that the number of $a$'s does not change while incrementing the 
second counter, it is sufficient to add two transitions labelled by $c^+_2$
with weight $-1$ in parallel of the ones labelled by $c^+_1$ in the automaton of Figure~\ref{figure:krob2}.

As for the decrement, it is the same idea, except that a special 
case needs to be considered when the value of the counter is already $0$,
as shown in Figure~\ref{figure:krob2bis}.

Similar automata are constructed to check the good behaviour of the number of $b$'s.

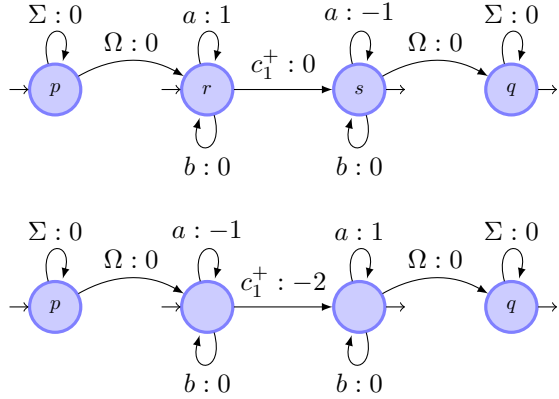
\begin{figure}[!htbp]
\begin{center}
\begin{tikzpicture}[scale=0.5]
\tikzset{
	every state/.style={draw=blue!50,very thick,fill=blue!20,scale=0.7},
	initial text=,
	accepting/.style={accepting by arrow},
	fleche/.style={->, >=latex}
}

\node[state, initial] (q_1) at (-6,0) {$p$};
\node[state, initial] (q_2) at (-2,0) {$r$};
\node[state,accepting] (q_3) at (2,0) {$s$};
\node[state, accepting] (q_4) at (6,0) {$q$};

\path[fleche]     (q_1) edge [loop above] node         {\small{$\alphabet:0$}} ();
\path[fleche]     (q_1) edge [bend left] node [above] {\small{$\Omega:0$}}(q_2);
\path[fleche]     (q_2) edge [loop above] node         {\small{$a:1$}} ();
\path[fleche]     (q_2) edge [loop below] node         {\small{$b:0$}} ();
\path[fleche]     (q_2) edge  node [above] {\small{$c_1^{+}:0$}}(q_3);
\path[fleche]     (q_3) edge [loop above] node         {\small{$a:-1$}} ();
\path[fleche]     (q_3) edge [loop below] node         {\small{$b:0$}} ();
\path[fleche]     (q_3) edge [bend left] node [above] {\small{$\Omega:0$}}(q_4);
\path[fleche]     (q_4) edge [loop above] node         {\small{$\alphabet:0$}} ();
\end{tikzpicture}
\end{center}
\begin{center}
\begin{tikzpicture}[scale=0.5]
\tikzset{
	every state/.style={draw=blue!50,very thick,fill=blue!20,scale=0.7},
	initial text=,
	accepting/.style={accepting by arrow},
	fleche/.style={->, >=latex}
}

\node[state, initial] (q_1) at (-6,0) {$p$};
\node[state, initial] (q_2) at (-2,0) {};
\node[state, accepting] (q_3) at (2,0) {};
\node[state, accepting] (q_4) at (6,0) {$q$};

\path[fleche]     (q_1) edge [loop above] node         {\small{$\alphabet:0$}} ();
\path[fleche]     (q_1) edge [bend left] node [above] {\small{$\Omega:0$}}(q_2);
\path[fleche]     (q_2) edge [loop above] node         {\small{$a:-1$}} ();
\path[fleche]     (q_2) edge [loop below] node         {\small{$b:0$}} ();
\path[fleche]     (q_2) edge  node [above] {\small{$c_1^{+}:-2$}}(q_3);
\path[fleche]     (q_3) edge [loop above] node         {\small{$a:1$}} ();
\path[fleche]     (q_3) edge [loop below] node         {\small{$b:0$}} ();
\path[fleche]     (q_3) edge [bend left] node [above] {\small{$\Omega:0$}}(q_4);
\path[fleche]     (q_4) edge [loop above] node         {\small{$\alphabet:0$}} ();
\end{tikzpicture}
\end{center}
\caption{\label{figure:krob2} Automaton that checks if the first counter is 
well incremented.}
\end{figure}

\begin{figure}[!htbp]
\begin{center}
\begin{tikzpicture}[scale=0.5]
\tikzset{
	every state/.style={draw=blue!50,very thick,fill=blue!20,scale=0.7},
	initial text=,
	accepting/.style={accepting by arrow},
	fleche/.style={->, >=latex}
}

\node[state, initial] (q_1) at (-6,0) {$p$};
\node[state, initial] (q_2) at (-2,0) {};
\node[state,accepting] (q_3) at (2,0) {};
\node[state, accepting] (q_4) at (6,0) {$q$};

\path[fleche]     (q_1) edge [loop above] node         {\small{$\alphabet:0$}} ();
\path[fleche]     (q_1) edge [bend left] node [above] {\small{$\Omega:0$}}(q_2);
\path[fleche]     (q_2) edge [loop above] node         {\small{$b:0$}} ();
\path[fleche]     (q_2) edge  node [above] {\small{$c_1^{-}:-1$}}(q_3);
\path[fleche]     (q_3) edge [loop below] node         {\small{$a:1$}} ();
\path[fleche]     (q_3) edge [loop above] node         {\small{$b:0$}} ();
\path[fleche]     (q_3) edge [bend left] node [above] {\small{$\Omega:0$}}(q_4);
\path[fleche]     (q_4) edge [loop above] node         {\small{$\alphabet:0$}} ();
\end{tikzpicture}
\end{center}
\begin{center}
\begin{tikzpicture}[scale=0.5]
\tikzset{
	every state/.style={draw=blue!50,very thick,fill=blue!20,scale=0.7},
	initial text=,
	accepting/.style={accepting by arrow},
	fleche/.style={->, >=latex}
}

\node[state, initial] (q_1) at (-6,0) {$p$};
\node[state] (q_2) at (-2,0) {$r$};
\node[state,accepting] (q_3) at (2,0) {$s$};
\node[state, accepting] (q_4) at (6,0) {$q$};
\node[state, initial] (q) at (-4,-3) {};

\path[fleche]     (q_1) edge [loop above] node         {\small{$\alphabet:0$}} ();
\path[fleche]     (q_1) edge node [sloped, above] {\small{$\Omega:0$}}(q);
\path[fleche]     (q) edge  node [sloped, above] {\small{$a:1$}}(q_2);
\path[fleche]     (q_2) edge [loop above] node         {\small{$a:1$}} ();
\path[fleche]     (q_2) edge [loop below] node         {\small{$b:0$}} ();
\path[fleche]     (q_2) edge  node [above] {\small{$c_1^{-}:-2$}}(q_3);
\path[fleche]     (q_3) edge [loop above] node         {\small{$a:-1$}} ();
\path[fleche]     (q_3) edge [loop below] node         {\small{$b:0$}} ();
\path[fleche]     (q_3) edge [bend left] node [above] {\small{$\Omega:0$}}(q_4);
\path[fleche]     (q_4) edge [loop above] node         {\small{$\alphabet:0$}} ();
\end{tikzpicture}
\end{center}
\begin{center}
\begin{tikzpicture}[scale=0.5]
\tikzset{
	every state/.style={draw=blue!50,very thick,fill=blue!20,scale=0.7},
	initial text=,
	accepting/.style={accepting by arrow},
	fleche/.style={->, >=latex}
}

\node[state, initial] (q_1) at (-6,0) {$p$};
\node[state] (q_2) at (-2,0) {};
\node[state,accepting] (q_3) at (2,0) {};
\node[state, accepting] (q_4) at (6,0) {$q$};
\node[state, initial] (q) at (-4,-3) {};

\path[fleche]     (q_1) edge [loop above] node         {\small{$\alphabet:0$}} ();
\path[fleche]     (q_1) edge node [sloped, above] {\small{$\Omega:0$}}(q);
\path[fleche]     (q) edge node [sloped, above] {\small{$a:-1$}}(q_2);
\path[fleche]     (q_2) edge [loop above] node         {\small{$a:-1$}} ();
\path[fleche]     (q_2) edge [loop below] node         {\small{$b:0$}} ();
\path[fleche]     (q_2) edge  node [above] {\small{$c_1^{-}:0$}}(q_3);
\path[fleche]     (q_3) edge [loop above] node         {\small{$a:1$}} ();
\path[fleche]     (q_3) edge [loop below] node         {\small{$b:0$}} ();
\path[fleche]     (q_3) edge [bend left] node [above] {\small{$\Omega:0$}}(q_4);
\path[fleche]     (q_4) edge [loop above] node         {\small{$\alphabet:0$}} ();
\end{tikzpicture}
\end{center}
\caption{\label{figure:krob2bis} Automaton that checks if the first counter is 
well decremented.}
\end{figure}
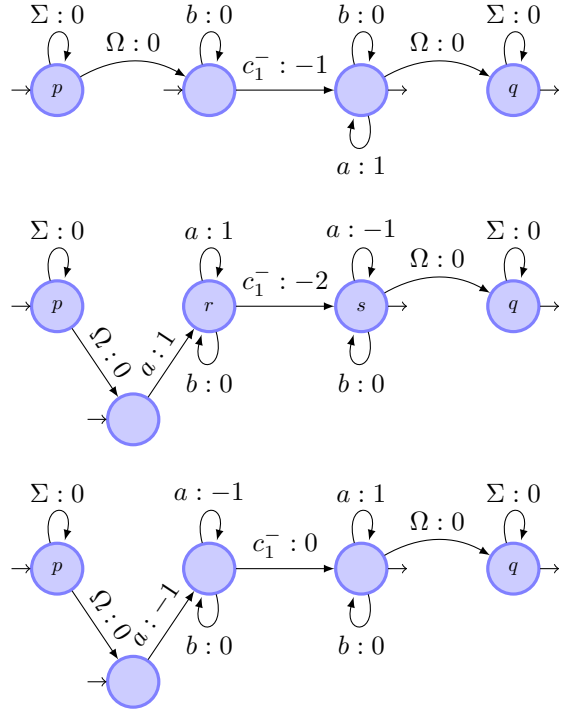


Finally, for item \ref{an}., we construct an automaton $\A_3$
as in Figure~\ref{figure:krob3}.
Consider a word $w$ in
$a^m\Omega\alphabet^*$. 
For the first part (above) of $\A_3$, the value computed on $w$ is $m-n-1$. 
As for the second part (below), the value computed on $w$ is $-m+n-1$. Hence,
$\fonc{\A_3}(w) = \max(m-n-1,-m+n-1)= |m-n|-1$. Thus, 
if $n=m$ then $\fonc{\A_3}(w) = -1$, otherwise $\fonc{\A_3}(w) \geq 0$.

\begin{figure}[!htbp]
\begin{minipage}{7cm}
\begin{center}
\begin{tikzpicture}[scale=0.5]
\tikzset{
	every state/.style={draw=blue!50,very thick,fill=blue!20,scale=0.7},
	initial text=,
	accepting/.style={accepting by arrow},
	fleche/.style={->, >=latex}
}

\node[state, initial] (q_1) at (-3,0) {};
\node[state, accepting] (q_2) at (3,0) {$q$};

\path[fleche]     (q_1) edge [loop above] node         {\small{$a:1$}} ();
\path[fleche]     (q_1) edge  node [above] {\small{$\Omega:-n-1$}}(q_2);
\path[fleche]     (q_2) edge [loop above] node         {\small{$\alphabet:0$}} ();
\end{tikzpicture}
\end{center}
\end{minipage}
\begin{minipage}{7cm}
\begin{center}
\begin{tikzpicture}[scale=0.5]
\tikzset{
	every state/.style={draw=blue!50,very thick,fill=blue!20,scale=0.7},
	initial text=,
	accepting/.style={accepting by arrow},
	fleche/.style={->, >=latex}
}

\node[state, initial] (p_1) at (-3,0) {};
\node[state, accepting] (p_2) at (3,0) {$q$};

\path[fleche]     (p_1) edge [loop above] node         {\small{$a:-1$}} ();
\path[fleche]     (p_1) edge  node [above] {\small{$\Omega:n-1$}}(q_2);
\path[fleche]     (p_2) edge [loop above] node         {\small{$\alphabet:0$}} ();
\end{tikzpicture}
\end{center}
\end{minipage}
\caption{\label{figure:krob3} Automaton that checks if the first counter is 
initialised to $n$.}
\end{figure}
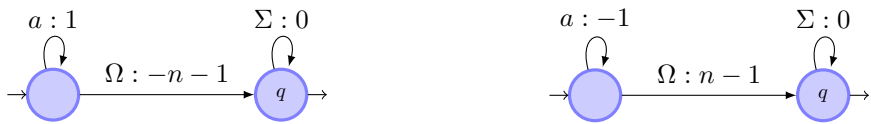


By merging states playing the same role, one can build a max-plus automaton 
$\A$ with $2|Q| + 27$ states computing the maximum of the functions described above. 
This max-plus automaton satisfies $\fonc{\A}(w) = -1$ if and only if $w$ 
represents a valid 
execution for $\mathcal{M}$ with counters initialised to $(n,0)$ and ending in
$q_{halt}$. Otherwise, $\fonc{\A}(w) \geq 0$.
\end{proof}

\begin{proof}[Proof of Theorem~\ref{theorem:main}]
The undecidability comes from a reduction from the halting problem of a U2CM, which is given in Lemma~\ref{lemma:transfoweights}.
There exists a U2CM with $268$ states \cite{Ivanov2014}, thus the comparison problem for max-plus automata with $553$ states on a $6$-letter alphabet
with weight in $\{-n-1,-2,-1,0,1,n-1\}$ is undecidable.
\end{proof}

\section{Conclusion and open questions}

In this paper, we have proved that the joint spectral radius and 
the ultimate rank of a finite set of matrices over the tropical semiring
are not computable
(from the proof it can be seen that they are actually computably-enumerable-complete).
To this end, we have proved the undecidability
of the comparison and equivalence of max-plus automata in restricted cases:
when all the states are both initial and final and when the number of states
is bounded.

As for the restriction on the number of states, we proved that comparison 
is undecidable when restricted to $553$ states. Now, the question is to understand what happens 
between $2$ and $552$ states. Even when restricted to $2$ states, it seems quite a difficult 
question to answer. Moreover, the various proofs highlight a link between several universal 
models: diophantine equations and two-counter machines. 
Having better size bounds on these models would give a better bound for
our undecidable problem, but conversely, getting the decidability of comparison 
for max-plus automata with at most a certain number of states could lead to 
improve the known lower bounds on the size of these universal objects.

As for the joint spectral radius, one
could ask if it is always rational or if, on the opposite, the set of 
joint spectral radii of finite families of matrices admits some computability-theoretic characterization.
With respect to complexity, the main open question is whether it
is PSPACE to approximate the joint spectral radius. 

Finally, 
we have the undecidability of comparison for max-plus automata whose states are all
both initial and final over an alphabet of fixed size. 
Standard techniques can encode any 
alphabet into a $2$-letter one, but it is unclear how to adapt them so as to maintain all states 
initial and final. Then the decidability of computing the joint spectral radius  
of a set of $2$ matrices is still open.

\bibliographystyle{plain}
\bibliography{urkind}

\begin{thebibliography}{10}

\bibitem{Almagor2011}
Shaull Almagor, Udi Boker, and Orna Kupferman.
\newblock What’s decidable about weighted automata?
\newblock In {\em ATVA 2011}, pages 482--491. Springer-Verlag, oct 2011.

\bibitem{RhoMinNPHard}
Vincent~D. Blondel, Stéphane Gaubert, and John~N. Tsitsiklis.
\newblock Approximating the spectral radius of sets of matrices in the
  max-algebra is np-hard.
\newblock {\em Automatic Control, IEEE Transactions on}, 45(9):1762--1765, Sep
  2000.

\bibitem{Colcombet}
Thomas Colcombet.
\newblock On distance automata and regular cost function.
\newblock Presented at the Dagstuhl seminar “Advances and Applications of
  Automata on Words and Trees”, 2010.

\bibitem{ColDav13}
Thomas Colcombet and Laure Daviaud.
\newblock Approximate comparison of distance automata.
\newblock In Natacha Portier and Thomas Wilke, editors, {\em STACS}, volume~20
  of {\em LIPIcs}, pages 574--585. Schloss Dagstuhl - Leibniz-Zentrum fuer
  Informatik, 2013.

\bibitem{ColDav14}
Thomas Colcombet, Laure Daviaud, and Florian Zuleger.
\newblock Size-change abstraction and max-plus automata.
\newblock In Erzs{\'{e}}bet Csuhaj{-}Varj{\'{u}}, Martin Dietzfelbinger, and
  Zolt{\'{a}}n {\'{E}}sik, editors, {\em Mathematical Foundations of Computer
  Science 2014 - 39th International Symposium, {MFCS} 2014, Budapest, Hungary,
  August 25-29, 2014. Proceedings, Part {I}}, volume 8634 of {\em Lecture Notes
  in Computer Science}, pages 208--219. Springer, 2014.

\bibitem{GaubertKatz2006}
St\'ephane Gaubert and Ricardo Katz.
\newblock Reachability problems for products of matrices in semirings.
\newblock {\em International Journal of Algebra and Computation},
  16(3):603--627, jun 2006.

\bibitem{GaubertMairesse98}
St{\'e}phane Gaubert and Jean Mairesse.
\newblock Task resource models and {$(\max,+)$} automata.
\newblock In {\em Idempotency ({B}ristol, 1994)}, volume~11 of {\em Publ.
  Newton Inst.}, pages 133--144. Cambridge Univ. Press, Cambridge, 1998.

\bibitem{Gaub95}
Stéphane Gaubert.
\newblock Performance evaluation of {$(\max,+)$} automata.
\newblock {\em IEEE Trans. Automat. Control}, 40(12):2014--2025, 1995.

\bibitem{Gaubert96Burnside}
Stéphane Gaubert.
\newblock On the {B}urnside problem for semigroups of matrices in the (max,+)
  algebra.
\newblock {\em Semigroup Forum}, 52(1):271--294, 1996.

\bibitem{GaubMair99}
Stéphane Gaubert and Jean Mairesse.
\newblock Modeling and analysis of timed {P}etri nets using heaps of pieces.
\newblock {\em IEEE Trans. Automat. Control}, 44(4):683--697, 1999.

\bibitem{urk}
Pierre Guillon, Zur Izhakian, Jean Mairesse, and Glenn Merlet.
\newblock The ultimate rank of semi-groups of tropical matrices.
\newblock {\em Journal of Algebra}, 437:222--248, September 2015.

\bibitem{Ivanov2014}
Sergiu Ivanov.
\newblock {\em On the Power and Universality of Biologically-inspired Models of
  Computation}.
\newblock PhD thesis, Université Paris-Est, 2014.

\bibitem{jones82}
James~P. Jones.
\newblock Universal {D}iophantine equation.
\newblock {\em J. Symbolic Logic}, 47(3):549--571, 1982.

\bibitem{JSRBookJungers}
Raphaël Jungers.
\newblock {\em The joint spectral radius}, volume 385 of {\em Lecture Notes in
  Control and Information Sciences}.
\newblock Springer-Verlag, Berlin, 2009.
\newblock Theory and applications.

\bibitem{Krob92}
Daniel Krob.
\newblock The equality problem for rational series with multiplicities in the
  tropical semiring is undecidable.
\newblock In {\em Automata, languages and programming ({V}ienna, 1992)}, volume
  623 of {\em Lecture Notes in Comput. Sci.}, pages 101--112. Springer, Berlin,
  1992.

\bibitem{LombardyMairesseSurvey}
Sylvain Lombardy and Jean Mairesse.
\newblock Max-plus automaton.
\newblock In {\em Handbook of Automata}. European Mathematical Society, To
  appear.

\bibitem{MERLET2010}
Glenn Merlet.
\newblock Semigroup of matrices acting on the max-plus projective space.
\newblock {\em Linear Algebra and its Applications}, 432(8):1923 -- 1935, 2010.

\bibitem{Minsky61}
Marvin~L. Minsky.
\newblock Recursive unsolvability of {P}ost's problem of ``tag'' and other
  topics in theory of {T}uring machines.
\newblock {\em Ann. of Math. (2)}, 74:437--455, 1961.

\bibitem{Minsky67}
Marvin~L. Minsky.
\newblock {\em Computation: finite and infinite machines}.
\newblock Prentice-Hall, Inc., Englewood Cliffs, N.J., 1967.
\newblock Prentice-Hall Series in Automatic Computation.

\bibitem{Schutz61}
Marcel-Paul Sch{\"u}tzenberger.
\newblock On the definition of a family of automata.
\newblock {\em Information and Control}, 4:245--270, 1961.

\end{thebibliography}

\end{document}